\documentclass[conference,10pt]{IEEEtran}

\usepackage{amsmath,graphicx}
\usepackage{amsfonts}
\usepackage{amssymb}
\usepackage{epsfig}
\usepackage{mathrsfs}
\usepackage{graphicx}
\usepackage{algorithm}
\usepackage{caption}
\usepackage{subcaption}
\usepackage{epstopdf}
\usepackage{balance}

\newtheorem{theorem}{Theorem}[section]

\newenvironment{proof}[1][Proof]{\begin{trivlist}
\item[\hskip \labelsep {\bfseries #1}]}{\end{trivlist}}

\newcommand\qed{$\blacksquare$}

\DeclareMathOperator{\trace}{\mathrm{tr} \, }
\DeclareMathOperator{\diag}{\mathrm{diag} \, }

\newsavebox\myboxA
\newsavebox\myboxB
\newlength\mylenA

\newcommand{\oline}[1]{\mkern 1.5mu\overline{\mkern-1.5mu#1\mkern-1.5mu}\mkern 1.5mu}

\def\mA{\mbox{$\mathbf{A}$}}
\def\mB{\mbox{$\mathbf{B}$}}
\def\mBc{\mbox{$\oline{\mathbf{B}}$}}

\def\mD{\mbox{$\mathbf{D}$}}
\def\mDc{\mbox{$\oline{\mathbf{D}}$}}
\def\mG{\mbox{$\mathbf{G}$}}

\def\mI{\mbox{$\mathbf{I}$}}
\def\mL{\mbox{$\mathbf{L}$}}

\def\mU{\mbox{$\mathbf{U}$}}

\def\mSigma{\mbox{$\mathbf{\Sigma} \kern .08em$}}
\def\mLambda{\mbox{$\mathbf{\Lambda} \kern .08em$}}

\newcommand{\A}{{\cal A}}
\newcommand{\E}{{\cal E}}
\newcommand{\G}{{\cal G}}

\newcommand{\V}{{\cal V}}

\newcommand{\abs}[1]{\left\vert#1\right\vert}

\def\S{\text{\mbox{${\cal S}$}}}
\def\F{\text{\mbox{${\cal F}$}}}

\def\V{\text{\mbox{${\cal V}$}}}
\def\E{\text{\mbox{${\cal E}$}}}
\def\G{\text{\mbox{${\cal G}$}}}
\def\B{\text{\mbox{${\cal B}$}}}
\def\D{\text{\mbox{${\cal D}$}}}

\def\bG{\text{\mbox{\boldmath $G$}}}

\def\bL{\text{\mbox{\boldmath $L$}}}

\def\b0{\text{\mbox{\boldmath $0$}}}

\def\bff{\text{\mbox{\boldmath $f$}}}
\def\bffs{\text{\mbox{\boldmath $f_{\cal S}$}}}

\def\bn{\text{\mbox{\boldmath $n$}}}

\def\br{\text{\mbox{\boldmath $r$}}}
\def\bs{\text{\mbox{\boldmath $s$}}}

\def\bx{\text{\mbox{\boldmath $x$}}}

\def\buno{\text{\mbox{\boldmath $1$}}}

\def\bpsi{\text{\mbox{\boldmath $\psi$}}}

\begin{document}

\title{Uncertainty Principle and Sampling \\ of Signals Defined on Graphs}

%

\author{Mikhail~Tsitsvero$^1$, Sergio~Barbarossa$^1$, and Paolo Di Lorenzo$^2$\\
\small $^1$Department of Information Eng., Electronics and Telecommunications, Sapienza University of Rome,\\
\small $^2$Department of Engineering, University of Perugia, Via G. Duranti 93, 06125, Perugia, Italy,\\
\small E-mail: {\tt tsitsvero@gmail.com, sergio.barbarossa@uniroma1.it, paolo.dilorenzo@unipg.it}}

\maketitle

\begin{abstract}
In many applications of current interest, the observations are represented as a signal defined over
a graph. The analysis of such signals requires the extension of standard signal processing tools.
Building on the recently introduced Graph Fourier Transform, the first contribution of this paper is
to provide an uncertainty principle for signals on graph. As a by-product of this theory, we
show how to build a dictionary of maximally concentrated signals on vertex/frequency domains.
Then, we establish a direct relation between  uncertainty principle and sampling, which forms
the basis for a sampling theorem of signals defined on graph. Based on this theory, we show that, besides sampling rate,
the samples' location plays a key role in the performance of signal recovery algorithms. Hence, we suggest
 a few alternative sampling strategies and compare them with recently proposed methods.
\end{abstract}


\section{Introduction}
\label{sec:intro}
\IEEEPARstart{I}{n} many applications, from sensor to social networks, gene regulatory networks or big data, observations can be represented as a signal defined over the vertices of a graph \cite{shuman2013emerging}, \cite{sandryhaila2014big}. Over the last few years, a series of papers produced a significant advancement in the development of processing tools for the analysis of signals defined over a graph, or graph signals for short \cite{shuman2013emerging}, \cite{sandryhaila2013discrete}. A central role is of course played by spectral analysis of graph signals, which passes through the introduction of the so called Graph Fourier Transform (GFT). Alternative definitions of GFT exist, depending on the different perspectives used to extend classical tools. Two basic approaches are available, proposing the projection of the graph signal onto the eigenvectors of either the graph Laplacian, see, e.g., \cite{shuman2013emerging}, \cite{zhu2012approximating} or of the adjacency matrix, see, e.g. \cite{sandryhaila2013discrete}, \cite{chen2015discrete}. Typically, even though a Laplacian matrix can be defined for both directed and undirected graphs, the methods in the first class assume undirected graphs, whereas the methods in the second class consider the more general directed case. Given the GFT definition, in \cite{agaskar2013spectral} and very recently in \cite{pasdeloup2015toward}, \cite{benedettograph}, \cite{koprowski2015finite}, it was derived a graph uncertainty principle aimed at expressing the fundamental relation between the spread of a signal over the vertex and spectral domains. The approach used in \cite{agaskar2013spectral} is based on the transposition of classical Heisenberg's method to graph signals.  However, although the results are interesting, this transposition gives rise to a series of questions, essentially related to the fact that while time and frequency domains are inherently metric spaces, the vertex domain is not. This requires a careful reformulation of spread in vertex and its transformed domain, which should not make any assumption about ordering and metrics over the graph domain.

A further fundamental tool in signal processing is sampling theory. An initial basic contribution to the extension of sampling theory to graph signals was given in \cite{pesenson2008sampling}. The theory developed in \cite{pesenson2008sampling} aimed to show that, given a subset of samples, there exists a cutoff frequency $\omega$ such that, if the spectral support of the signal lies in $[0, \omega]$, the overall signal can be reconstructed with no errors. Later, \cite{narang2013signal} extended  the results of \cite{pesenson2008sampling} providing a method to identify uniqueness sets, compute the cut-off frequency and to interpolate signals which are not exactly band-limited. Further very recent works  provided the conditions for perfect recovery of band-limited graph signals: \cite{chenicassp2015}, \cite{chen2015discrete}, based on the adjacency matrix formulation of the GFT; \cite{TsitsveroEusipco15}, based on the identification of an orthonormal basis maximally concentrated over the joint vertex/frequency domain; \cite{wang2014local}, based on local-set graph signal reconstructions; \cite{marquez2015}, illustrating the conditions for perfect recovery, based on successive local aggregations.

The contribution of this paper is threefold: a) we derive an uncertainty principle for graph signals, based on the generalization of classical Slepian-Landau-Pollack seminal works \cite{Slepian:1961:PSW}, \cite{landau1961prolate}, including the conditions for perfect localization of a graph signal in {\it both} vertex and frequency domains; b) we establish a link between uncertainty principle and sampling theory, thus deriving the necessary and sufficient conditions for the recovery of band-limited graph signals from its samples; c) we provide alternative sampling strategies aimed at improving the performance of the recovery algorithms in the presence of noisy observations and compare their performance with recently proposed methods.

\section{Basic Definitions}

We consider a graph $\G = (\V, \E)$ consisting of a set of $N$ nodes $\V = \{1,2,..., N\}$, along with a set of weighted edges $\E=\{a_{ij}\}_{i, j \in \V}$, such that $a_{ij}>0$, if there is a link from node $j$ to node $i$, or $a_{ij}=0$, otherwise.
A signal $\bx$ over a graph $\G$ is defined as a mapping from the vertex set to complex vectors of size $N$, i.e. $\bx : \V \rightarrow \mathbb{C}^{\abs{\V}} $.  The adjacency matrix $\mA$ of a graph is the collection of all the weights $a_{ij}, i, j = 1, \ldots, N$.  The degree of node $i$ is $k_i:=\sum_{j=1}^{N}a_{ij}$. The degree matrix is a diagonal matrix having the node degrees on its diagonal: $\mathbf{K} = \diag \{ k_1, k_2, ... , k_N \}$. The combinatorial Laplacian matrix is defined as $\mathbf{L} = \mathbf{K}-\mathbf{A}$.
If the graph is undirected, the Laplacian matrix is symmetric and positive semi-definite, and admits the eigendecomposition $\mathbf{L}=\mathbf{U}\boldsymbol{\Lambda}\mathbf{U}^H$, where $\mathbf{U}$ collects all the eigenvectors of $\mathbf{L}$ in its columns, whereas $\boldsymbol{\Lambda}$ is a diagonal matrix containing the eigenvalues of $\mathbf{L}$. The Graph Fourier Transform (GFT) has been defined in alternative ways, see, e.g., \cite{shuman2013emerging}, \cite{zhu2012approximating}, \cite{sandryhaila2013discrete}, \cite{chen2015discrete}. In this paper, we follow the approach based on the Laplacian matrix, but the theory can be extended to the adjacency based approach with minor modifications. In the Laplacian-based approach,
the Graph Fourier Transform $\mathbf{\hat{\bx}}$ of a vector $\bx$ is defined as the projection of  $\bx$ onto the space spanned by the eigenvectors of $\bL$ \cite{shuman2013emerging}, i.e.
\begin{equation}
\mathbf{\hat{\bx}} = \mathbf{U}^H \bx,
\end{equation}
where $^H$ denotes Hermitian operator (conjugate and transpose). The inverse Fourier transform is then
\begin{equation}
\mathbb{\bx} = \mathbf{U} \mathbf{\hat{\bx}}.
\end{equation}
Given a subset of vertices $\S \subseteq \V$, we define a vertex-limiting operator as the diagonal matrix
\begin{equation}
\label{D}
\mathbf{D}_{\S} = {\rm Diag}\{\buno_{\S}\},
\end{equation}
where $\buno_{\S}$ is the set indicator vector, whose $i$-th entry is equal to one, if  $i \in \S$, or zero otherwise. Similarly, given a subset of frequency indices $\F \subseteq \V^*$, we introduce the filtering operator
\begin{equation}
\label{lowpass_operator}
\mB_{\F} = \mathbf{U \Sigma_{\F} U}^H,
\end{equation}
where $\mathbf{\Sigma_{\F}}$ is a diagonal matrix defined as $\mathbf{\Sigma_{\F}}= {\rm Diag}\{\buno_{\F}\}$.
It is immediate to check that both matrices $\mathbf{D}_{\S}$ and $\mathbf{B}_{\F}$ are self-adjoint and idempotent, and then they represent orthogonal projectors. We refer to the space of all signals whose GFT is exactly supported on the set $\F$,  as the {\it Paley-Wiener space} for the set $\F$. We denote by $\B_{\F} \subseteq L_2 (\G)$ the set of all finite $\ell_2$-norm signals belonging to the Paley-Wiener space associated to $\F$. Similarly, we denote by $\D_{\S} \subseteq L_2(\G)$ the set of all finite $\ell_2$-norm signals with support on the vertex subset $\S$. In the rest of the paper, whenever there will be no ambiguities in the specification of the sets, we will drop the subscripts referring to the sets. Finally, given a set $\S$, we denote its complement set as $\oline{\S}$, such that $\V=\S \cup \oline{\S}$ and $\S \cap \oline{\S}=\emptyset$. Correspondingly, we define the vertex-projector onto $\oline{\S}$ as $\oline{\mathbf{D}}$ and, similarly, the frequency projector onto the frequency domain $\oline{\F}$ as $\oline{\mathbf{B}}$.

\section{Localization Properties}
\label{Localization}

In this section we derive the class of signals maximally concentrated over given subsets $\S$ and $\F$ in vertex and frequency domains.
We say that a vector $\bx$ is perfectly localized over the subset $\S \subseteq \V$ if
\begin{equation}
\label{Dx=x}
\mD \bx=\bx,
\end{equation}
with $\mD$ defined as in (\ref{D}).
Similarly, a vector $\bx$ is perfectly localized over the frequency set $\F \subseteq \V^*$ if
\begin{equation}
\label{Bx=x}
\mB \bx=\bx,
\end{equation}
with $\mB$ given in (\ref{lowpass_operator}).
Differently from continuous-time signals, a graph signal can be perfectly localized in {\it both} vertex and frequency domains. This
is stated in the following theorem.

\begin{theorem}
\label{theorem_unit_eigenvalue}
There exists a vector $\bx \in L_2 (\G)$ that is perfectly localized over both vertex set $\S$ and frequency set $\F$ if and only if the operator $\mB \mD \mB$ has an eigenvalue equal to one; in such a case, $\bx$ is an eigenvector associated to the unit eigenvalue.
\end{theorem}
\begin{proof}
The proof can be found in \cite{tsitsvero2015signals}. \hspace{3cm} \qed
\end{proof}
Equivalently, the perfect localization properties can be expressed in terms of the operators $\mB \mD$ and $\mD \mB$, thus leading to the following condition:
\begin{equation}
\label{|BD|=1=|DB|}
\|\mB \mD\|_2 = 1; \,\,\,\, \|\mD \mB\|_2 = 1.
\end{equation}

\noindent Typically, given two generic domains $\S$ and $\F$, we might not have perfectly concentrated signals in both domains. In such a case, it is worth finding the class of signals with limited support in one domain and maximally concentrated on the other. For example, we may search for the class of perfectly band-limited signals, i.e. $\mB \bx = \bx$, which are maximally concentrated in a vertex domain $\S$ or, viceversa, the class of signals with support on a subset of vertices, i.e. $\mD \bx = \bx$, which are maximally concentrated in a frequency  domain $\F$. The following theorem introduces the class of maximally concentrated functions in the band-limited scenario.

\begin{theorem}
\label{theorem::max_concentrated_vectors}
The class of orthonormal band-limited vectors $\bpsi_i$, $i=1, \ldots, N$, with $\mB \bpsi_i=\bpsi_i$, maximally concentrated over a vertex set $\S$, is given by the eigenvectors of $\mB \mD \mB$, i.e.
\begin{equation}
\label{slep_func:bdb_eigendecomposition}
\mathbf{BDB} \bpsi_i = \lambda_i \bpsi_i,
\end{equation}
with $\lambda_1\ge \lambda_2 \ge \ldots \ge \lambda_N$.
Furthermore, these vectors are orthogonal over the set $\S$, i.e. $\langle \bpsi_i, \mathbf{D} \bpsi_j \rangle = \lambda_j \delta_{ij}$, where $\delta_{ij}$ is the Kronecker symbol.
\end{theorem}
\begin{proof}
The proof can be found in \cite{tsitsvero2015signals}. \hspace{3cm}\qed
\end{proof}
The vectors $\bpsi_i$ are the counterpart of the prolate spheroidal wave functions introduced by Slepian and Pollack in \cite{Slepian:1961:PSW}.

\section{Uncertainty principle}
\label{Uncertainty principle}
A cornerstone property of continuous-time signals is the Heisenberg's principle, stating that a signal cannot be perfectly localized in both time and frequency domains, see, e.g., \cite{Folland1997}. More specifically, given a continuous-time signal $x(t)$ centered around $t_0$ and its Fourier Transform $X(f)$ centered at $f_0$, the uncertainty principle states that $\Delta_t^2 \Delta_f^2 \ge \frac{1}{(4 \pi)^2}$, where $\Delta_t^2$ and $\Delta_f^2$ are computed as the second order moments of the instantaneous power $|x(t)|^2$ and the spectral density $|X(f)|^2$, centered around their center of gravity points $t_0$ and $f_0$, respectively.
Quite recently, the uncertainty principle was extended to signals on graphs in \cite{agaskar2013spectral} by following an approach based on the transposition of the previous definitions of time and frequency spreads to graph signals. However, although interesting, this transposition hides a number of subtleties, which can limit the status of the result as a ``fundamental'' result, i.e. a result not constrained to any specific choice. More specifically, what happens is that the second order moments $\Delta_t^2$ and $\Delta_f^2$ contain a measure of distance in the time and frequency domains. When transposing these formulas to graph signals, it is necessary to define a distance between vertices of a graph. This is done in \cite{agaskar2013spectral} by using a common measure of graph distance, defined as the sum of weights along the shortest path between two vertices (equal to the number of hops, in case of unweighted graph). However, although perfectly legitimate, this formulation raises a number of questions: i) Is it correct, within the context of deriving fundamental limits, to exchange vertex or frequency distances with a graph distance defined as number of (possibly weighted) hops? ii) when moving by a time interval $\Delta t$ from $t_0$, we always get a single signal value $x(t)$; however, when moving $m$ steps away from a node $n_0$, we may encounter several nodes at the same distance; how should we weight these different contributions, possibly without making arbitrary assumptions? To overcome the above problems, in this paper we resort to an alternative definition of spread in vertex and frequency domain, generalizing the works of  Slepian, Landau and Pollack \cite{Slepian:1961:PSW}, \cite{landau1961prolate}. In particular, given a vertex set $\S$ and a frequency set $\F$, we denote by $\alpha^2$ and $\beta^2$ the percentage of energy falling within the sets $\S$ and $\F$, respectively, as
\begin{equation}
\frac{\|\mD \bx\|_2^2}{\| \bx\|_2^2}=\alpha^2;\qquad \frac{\|\mB \bx\|_2^2}{\| \bx\|_2^2}=\beta^2.
\end{equation}
Generalizing the approach of \cite{landau1961prolate} to graph signals, our goal is to find out the region of all admissible pairs $(\alpha, \beta)$ and to illustrate which are the graph signals able to attain all the points in such a region. The uncertainty principle is stated in the following theorem.

\begin{theorem}
\label{theorem::Uncertainty principle}
There exists $\bff \in L_2(\G)$ such that $\|\bff\|_2 = 1$, $\| \mD \bff \|_2 = \alpha$, $\| \mB \bff \|_2 = \beta $ if and only if $\left( \alpha, \beta \right) \in \Gamma$, with 
\begin{align}
\label{eq::uncertainty_region_Gamma}
\Gamma &= \left\{ \left( \alpha , \beta \right) : \right. \nonumber \\
& \cos^{-1} \alpha + \cos^{-1} \beta \geq \cos^{-1} \sigma_{max} \left( \mB \mD \right), \nonumber \\
& \cos^{-1} \sqrt{1 - \alpha^2} + \cos^{-1} \beta \geq \cos^{-1} \sigma_{max} \left( \mB \mDc \right), \\
& \cos^{-1} \alpha + \cos^{-1} \sqrt{ 1 - \beta^2} \geq \cos^{-1} \sigma_{max} \left( \mBc \mD \right), \nonumber \\
& \left. \cos^{-1} \sqrt{ 1- \alpha^2 } + \cos^{-1} \sqrt{1 - \beta^2} \geq \cos^{-1} \sigma_{max} \left( \mBc \mDc \right) \right\}. \nonumber
\end{align}
\end{theorem}
\begin{proof}
The proof can be found in \cite{tsitsvero2015signals}. \hspace{3cm}\qed
\end{proof}

An illustrative example of admissible region $\Gamma$ is reported in Fig. \ref{fig:uncertainty}. A few remarks about the border of the region $\Gamma$ are of interest. In general, any of the four curves at the corners of region $\Gamma$ in Fig. \ref{fig:uncertainty} may collapse onto the corresponding corner, whenever the conditions for perfect localization of the corresponding operator hold true. Furthermore, the curve in the upper right corner of $\Gamma$ specifies the pairs $(\alpha, \beta)$ that yield maximum concentration. This curve has equation
\begin{equation}
\label{eq::beta alpha curve}
\cos^{-1} \alpha + \cos^{-1} \beta = \cos^{-1} \sigma_{max}(\mathbf{B}\mathbf{D}).
\end{equation}
Solving with respect to $\beta$, and setting $\sigma^2_{max} := \sigma^2_{max}(\mathbf{B}\mathbf{D})$, we get
\begin{equation}
\label{beta vs alpha}
\beta = \alpha\, \sigma_{max} + \sqrt{(1-\alpha^2)(1-\sigma_{max}^2)}.
\end{equation}
For any given subset of nodes $\S$, as the cardinality of $\F$ increases, this upper curve gets closer and closer to the upper right corner, until it collapses on it, indicating perfect localization in both vertex and frequency domains. In particular, if we are interested in the allocation of energy within the sets $\S$ and $\F$ that maximizes, for example, the sum of the (relative) energies $\alpha^2+\beta^2$ falling in the vertex and frequency domains, the result is given by the intersection of the upper right curve, i.e. (\ref{beta vs alpha}), with the line $\alpha^2+\beta^2={\rm const}$. Given the symmetry of the curve (\ref{eq::beta alpha curve}), the result is achieved by setting $\alpha=\beta$, which yields
\begin{equation}
\alpha^2=\frac 1 2 (1+\sigma_{max}).
\end{equation}
The corresponding function $\bff'$ can then be written in closed form as:
\begin{equation}
\label{eq::opt_vectro_closed_form}
\bff' = \frac{\bpsi_1 - \mD \bpsi_1}{\sqrt{2\left(1 + \sigma_{max} \right)}} + \sqrt{\frac{1+\sigma_{max}}{2\sigma^2_{max}}} \mD \bpsi_1,
\end{equation}
where $\bpsi_1$ is the eigenvector of $\mB \mD \mB$ corresponding to $\sigma_{max}^2$  (please refer to \cite{tsitsvero2015signals} for details).

\begin{figure}[t]
\vspace{-2cm}
\begin{minipage}[t]{1.0\linewidth}
  \centering

 \centerline{\includegraphics[natwidth=127mm,natheight=101.6mm]{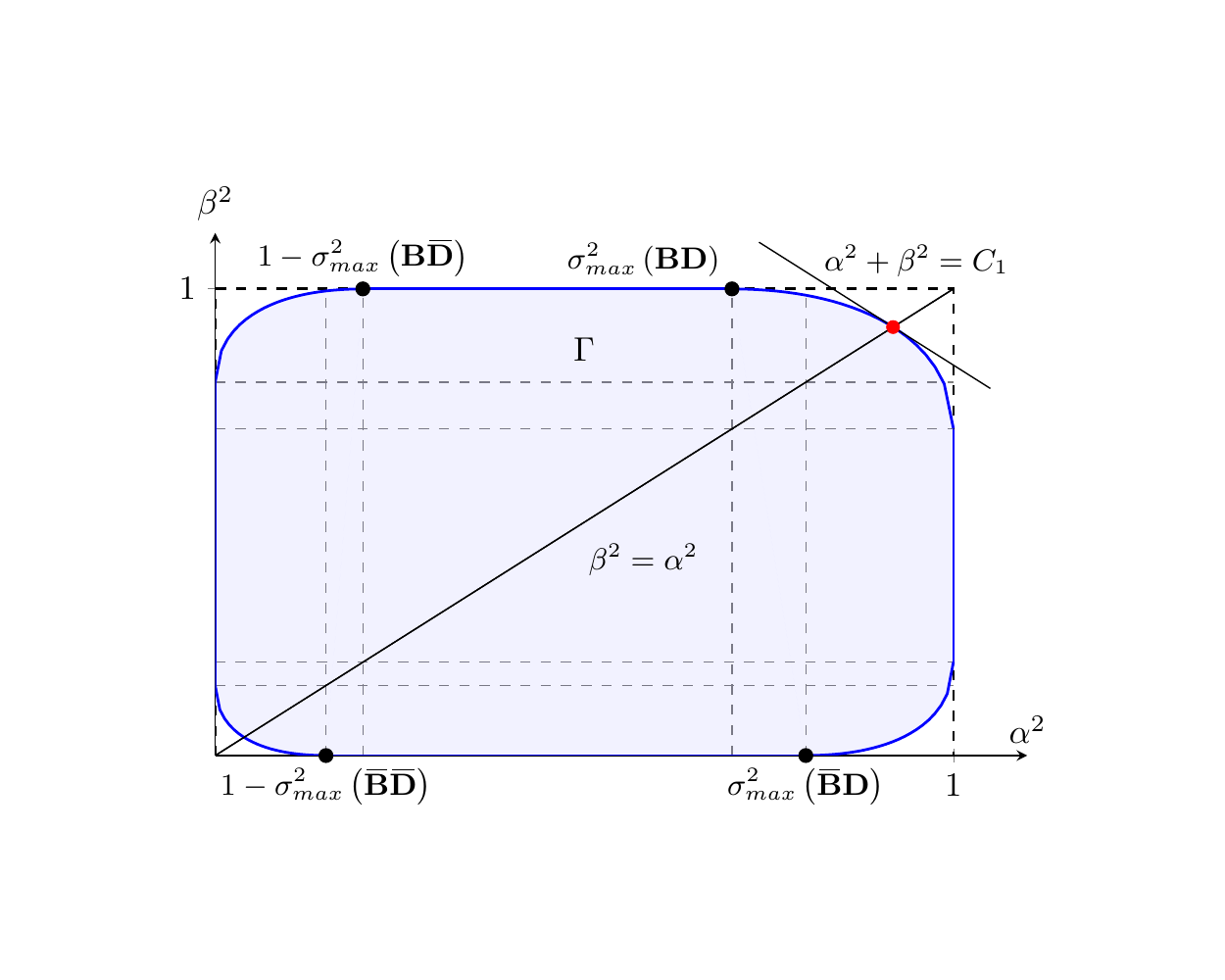}}
\end{minipage}
\vspace*{-2cm}
\caption{Admissible region $\Gamma$ of unit norm signals $\bff \in L^2(\G)$ with $\| \mD \bff \|_2=\alpha$ and $\| \mB \bff \|_2=\beta$.}
\label{fig:uncertainty}
\end{figure}

\section{Sampling}
\label{sec::Sampling}
Given a signal $\bff \in \B$ defined on the vertices of a graph, let us denote by $\bffs \in \D$ the vector equal to $\bff$ on a subset $\S \subseteq \V$ and zero outside:
\begin{equation}
\label{r=Ds}
\bffs := \mathbf{D} \bff.
\end{equation}
We wish to find out the conditions and the means for perfect recovery of $\bff$ from $\bffs$. The necessary and sufficient conditions are stated in the following sampling theorem.
\begin{theorem}[Sampling Theorem]
\label{theorem::sampling theorem}
Given a band-limited vector $\bff \in \B$, it is possible to recover $\bff$ from its samples taken from the set $\S$, if  and only if
\begin{equation}
\label{|DcB|<1}
\| \mB \mDc \|_2 < 1,
\end{equation}
i.e. if the matrix $\mB \mDc \mB$ does not have any eigenvector that is perfectly localized on $\oline{\S}$ and bandlimited on $\F$. Any signal $\bff \in \B$ can then be reconstructed from its sampled version $\bff_{\scriptsize \S} \in \D$ using the following reconstruction formula:
\begin{equation}
\label{eq::sampling_theorem_formula}
\bff = \sum_{i=1}^{\scriptsize \abs{\F}} \frac{1}{\sigma_i^2} \langle \bffs, \bpsi_i \rangle \bpsi_i,
\end{equation}
where $\left\{ \bpsi_i \right\}_{i = 1 .. N}$ and $\left\{ \sigma^2_i \right\}_{i = 1..N}$ are the eigenvectors and eigenvalues of $\mB \mD \mB$, respectively.
\end{theorem}
\begin{proof}
The proof can be found in \cite{tsitsvero2015signals}. \hspace{3cm}\qed
\end{proof}

\noindent Let us study now the implications of condition (\ref{|DcB|<1}) of Theorem \ref{theorem::sampling theorem} on the sampling strategy. To fulfill (\ref{|DcB|<1}), we need to guarantee that there exist no band-limited signals, i.e. $\mB \bx =\bx$, such that $\mB\mDc\bx=\bx$.
To make (\ref{|DcB|<1}) hold true, we must then ensure that $\mB \mDc \bx \neq \bx$ or, equivalently, $ \mDc \mB \bx \neq \bx$. Since
\begin{equation}
\mB \bx = \bx = \mD\mB \bx+\mDc \mB \bx,
\end{equation}
we need to guarantee that  $\mD\mB \bx \neq\b0$. Let us define now the $|\S| \times |\F|$ matrix $\mG$ as
\begin{equation}
\mG=
\left(
\begin{array}{llll}
u_{i_1}(j_1) & u_{i_2}(j_1) & \cdots & u_{i_{\tiny|\F|}}(j_1)\nonumber\\
\vdots & \vdots & \vdots & \vdots \nonumber\\
u_{i_1}(j_{\tiny|\S|}) & u_{i_2}(j_{\tiny|\S|}) & \cdots & u_{i_{\tiny|\F|}}(j_{\tiny|\S|})
\end{array}
\right)
\end{equation}
whose $\ell$-th column is the eigenvector of index $i_{\ell}$ of the Laplacian matrix (or any orthonormal set of basis vectors), sampled at the positions indicated by the indices $j_1, \ldots, j_{\tiny|\S|}$. It is easy to understand how condition (\ref{|DcB|<1}) is equivalent to require $\mG$ to be full column rank. Of course, a necessary condition for the existence of a non trivial vector $\bx$ satisfying $\mD\mB\bx \neq \b0$, and then enabling sampling theorem, is that $|\S| \ge |\F|$. However, this condition is not sufficient, because $\mG$ may loose rank, depending on graph topology and samples' location. As an extreme case,
if the graph is not connected, the vertices can be labeled so that the Laplacian (adjacency) matrix can be written as a block diagonal matrix, with a number of blocks equal to the number of connected components. Correspondingly, each eigenvector of $\mL$ (or $\mA$) can be expressed as a vector having all zero elements, except the entries corresponding to the connected component, which that eigenvector is associated to. This implies that, if there are no samples over the vertices corresponding to the non-null entries of the eigenvectors with index included in $\F$, $\mG$ looses rank. In principle, a signal defined over a disconnected graph can still be reconstructed from its samples, but only provided that the number of samples belonging to each connected component is at least equal to the number of eigenvectors with indices in ${\F}$ associated to that component. More generally, even if the graph is connected, there may easily occur situations where matrix $\mG$ is not rank-deficient, but it is ill-conditioned, depending on graph topology and samples' location. This suggests that the location of samples plays a key role in the performance of the reconstruction algorithm, as we will show in the next section.

\section{Sampling of noisy signal}
\label{sec::Sampling of noisy signal}
To assess the effect of ill-conditioning of matrix $\bG$ or, equivalently, the possibility for the spectral norm of $\mB \mDc$ to be very close to one on the reconstruction algorithms, we consider the reconstruction of band-limited signals from noisy samples. The observation model is
\begin{equation}
\label{r=D(s+n)}
\br = \mD \left( \bs + \bn \right),
\end{equation}
where $\bn $ is a noise vector. Applying (\ref{eq::sampling_theorem_formula}) to $\br$, the reconstructed signal $\widetilde{\bs}$ is
\begin{equation}
\label{eq::signal_plus_noise_expansion}
\widetilde{\bs} = \sum_{i=1}^{\abs{\F}} \frac{1}{\sigma_i^2} \langle \mD \bs, \bpsi_i \rangle \bpsi_i + \sum_{i=1}^{\abs{\F}} \frac{1}{\sigma_i^2} \langle \mD \bn, \bpsi_i \rangle \bpsi_i.
\end{equation}
Exploiting the orthonormality of $\bpsi_i$, the mean square error is
\begin{align}
\label{eq::expected_value_noise}
MSE &= \mathbb{E}\left\{\| \widetilde{\bs}-\bs \|_2^2 \right\} =\mathbb{E}\left\{ \sum_{i=1}^{\abs{\F}} \frac{1}{\sigma_i^4} \abs{\langle \mD \bn, \bpsi_i \rangle }^2 \right\} \nonumber \\
&=  \sum_{i=1}^{\abs{\F}} \frac{1}{\sigma_i^4}   \bpsi_i^H\mD \mathbb{E}\left\lbrace \bn \bn^H \right\rbrace \mD \bpsi_i.
\end{align}
In case of identically distributed uncorrelated noise, i.e. $\mathbb{E} \left\lbrace \bn \bn^H \right\rbrace=\beta_n^2 \mI$, we get
\begin{align}
\label{eq::mse_gaussian_uncorrelated}
MSE_G = \sum_{i=1}^{\abs{\F}} \frac{\beta^2_n}{\sigma_i^4}  \trace \left( \bpsi_i^H \mD \bpsi_i \right) = \beta^2_n \sum_{i=1}^{\abs{\F}} \frac{1}{\sigma_i^2}.
\end{align}
Since the non-null singular values of the Moore-Penrose left pseudo-inverse $\left( \mB \mD \right)^+$ are the inverses of singular values of $\mB \mD$, expression (\ref{eq::mse_gaussian_uncorrelated}) can be rewritten as:
\begin{equation}
\label{eq::mse_frobenius}
MSE_G = \beta^2_n \, \| \left(\mB \mD \mB \right)^+ \|_F.
\end{equation}
Then, a possible optimal sampling strategy consists in selecting the vertices that minimize the mean square error in (\ref{eq::mse_frobenius}).

\subsection{Sampling strategies}
When sampling graph signals, besides choosing the right number of samples, whenever possible it is also fundamental to have a strategy indicating {\it where} to sample, as the samples' location plays a key role in the performance of reconstruction algorithms. Building on the analysis of signal reconstruction algorithms in the presence of noise carried out earlier in this section, a possible strategy is to select the location in order to minimize the MSE. From (\ref{eq::mse_frobenius}), taking into account that
\begin{equation}
\lambda_i \left( \mB \mD \mB \right) = \sigma^2_i \left( \mB \mD \right) = \sigma^2_i \left( \mSigma \mU^H \mD \right),
\end{equation}
the problem is equivalent to selecting the right columns of the matrix $ \mSigma \mU^H$ according to some optimization criterion. This is a combinatorial problem and in the following we will provide some numerically efficient, albeit sub-optimal, greedy algorithms to tackle the problem. We will then compare the performance with the benchmark case corresponding to the combinatorial solution. We will denote by $\tilde{\mU}$ the matrix whose rows are the first $\abs{\F}$ rows of $\mU^H$; the symbol $\tilde{\mU}_{\A}$ denotes the matrix formed with the columns of $\tilde{\mU}$ belonging to set $\A$. The goal is to find the sampling set $\S$, which amounts to selecting the best, in some optimal sense, $|\S|$ columns of $\tilde{\mU}$.


\noindent {\it Greedy Selection - Minimization of Frobenius norm of $\left( \mSigma \mU^H \mD \right)^+$}: This strategy aims at minimizing the MSE in (\ref{eq::mse_frobenius}), assuming the presence of uncorrelated noise. We propose a greedy approach to tackle this selection problem. The resulting strategy is summarized in Algorithm 1.

\begin{algorithm}[h]
$\textit{Input Data}:$ $\tilde{\mU}$, the first $\abs{\F}$ rows of $\mU^H$;

\hspace{1.1cm} $\qquad M$, the number of samples.

$\textit{Output Data}:$ $\S$, the sampling set. \smallskip

$\textit{Function}:$ \hspace{.23cm} initialize $\S\equiv \emptyset$

\hspace{2 cm} while $|\S|<M$

\hspace{2.3cm} $\displaystyle s=\arg \min_j \;\;\sum_{i=1}^{\abs{\F}} \frac{1}{\sigma_i^2(\tilde{\mU}_{\S\cup\{j\}})}$;

\hspace{2.3cm} $\S \leftarrow \S \cup \{s\}$;

\hspace{2cm} end

\noindent \caption{\label{alg:Greedy1}\textbf{: Greedy selection based on minimum Frobenius norm of $\left( \mSigma \mU^H \mD \right)^+$}}
\end{algorithm}

\noindent {\it Greedy Selection - Maximization of the volume of the parallelepiped formed with the columns of $\tilde{\mU}$:} In this case, the strategy aims at selecting the set $\mathcal{S}$ of columns of the matrix $\tilde{\mU}$ that maximize the (squared) volume of the parallelepiped built with the selected columns of $\tilde{\mU}$ in $\mathcal{S}$. This volume can be computed as the determinant of the matrix $\tilde{\mU}_\mathcal{S}^H \tilde{\mU}_\mathcal{S}$, i.e. $|\tilde{\mU}_\mathcal{S}^H \tilde{\mU}_\mathcal{S}|$. The rationale underlying this approach is not only to choose the columns with largest norm, but also the vectors more orthogonal to each other.
Also in this case, we propose a greedy approach, as described in Algorithm 2. The algorithm starts including the column with the largest norm in $\tilde{\mU}$, and then it adds, iteratively, the columns having the largest norm and, at the same time, as orthogonal as possible to the vectors already in $\S$.

\begin{algorithm}[h]
$\textit{Input Data}:$ $\tilde{\mU}$, the first $\abs{\F}$ rows of $\mU^H$;

\hspace{1.1cm} $\qquad M$, the number of samples.

$\textit{Output Data}:$ $\S$, the sampling set. \smallskip

$\textit{Function}:$ \hspace{.23cm} initialize $\S\equiv \emptyset$

\hspace{2 cm} while $|\S|<M$

\hspace{2.3cm} $\displaystyle s=\arg \max_j \;\; \left|\tilde{\mU}_{\S\cup\{j\}}^H \tilde{\mU}_{\S\cup\{j\}}\right|$;

\hspace{2.3cm} $\S \leftarrow \S \cup \{s\}$;

\hspace{2cm} end

\protect\caption{\label{alg:Greedy2}\textbf{: Greedy selection based on maximum parallelepiped volume}}
\end{algorithm}

%
%
%
%
%
%
%
%
%
%
%
%
%
%
%
%
%
\noindent {\it Greedy Selection - Maximization of the Frobenius norm of $\mSigma \mU^H \mD$:} Finally, we propose a strategy that aims at selecting the columns of the matrix $\tilde{\mU}$ that maximize its Frobenius norm. Even if this strategy is not directly related to the optimization of the MSE in (\ref{eq::mse_gaussian_uncorrelated}), it leads to a very simple implementation. Although clearly sub-optimal, this choice  will be later shown to provide fairly good performance if the number of samples is sufficiently larger than the theoretical limit. The method works as follows:
\begin{equation}
\max_\S \;\|\tilde{\mU}\mD\|_2^2\;=\; \max_\S \;\sum_{i \in \S} \|(\tilde{\mU})_i\|^2_2.
\end{equation}
The optimal selection strategy simply consists in selecting the $M$ columns from $\tilde{\mU}$ with largest $\ell_2$-norm. \smallskip

\subsection{Numerical Results}

\begin{figure*}[htp]
\centering
\centering
\begin{subfigure}[t]{0.45\linewidth}
\includegraphics[width=\columnwidth,keepaspectratio]{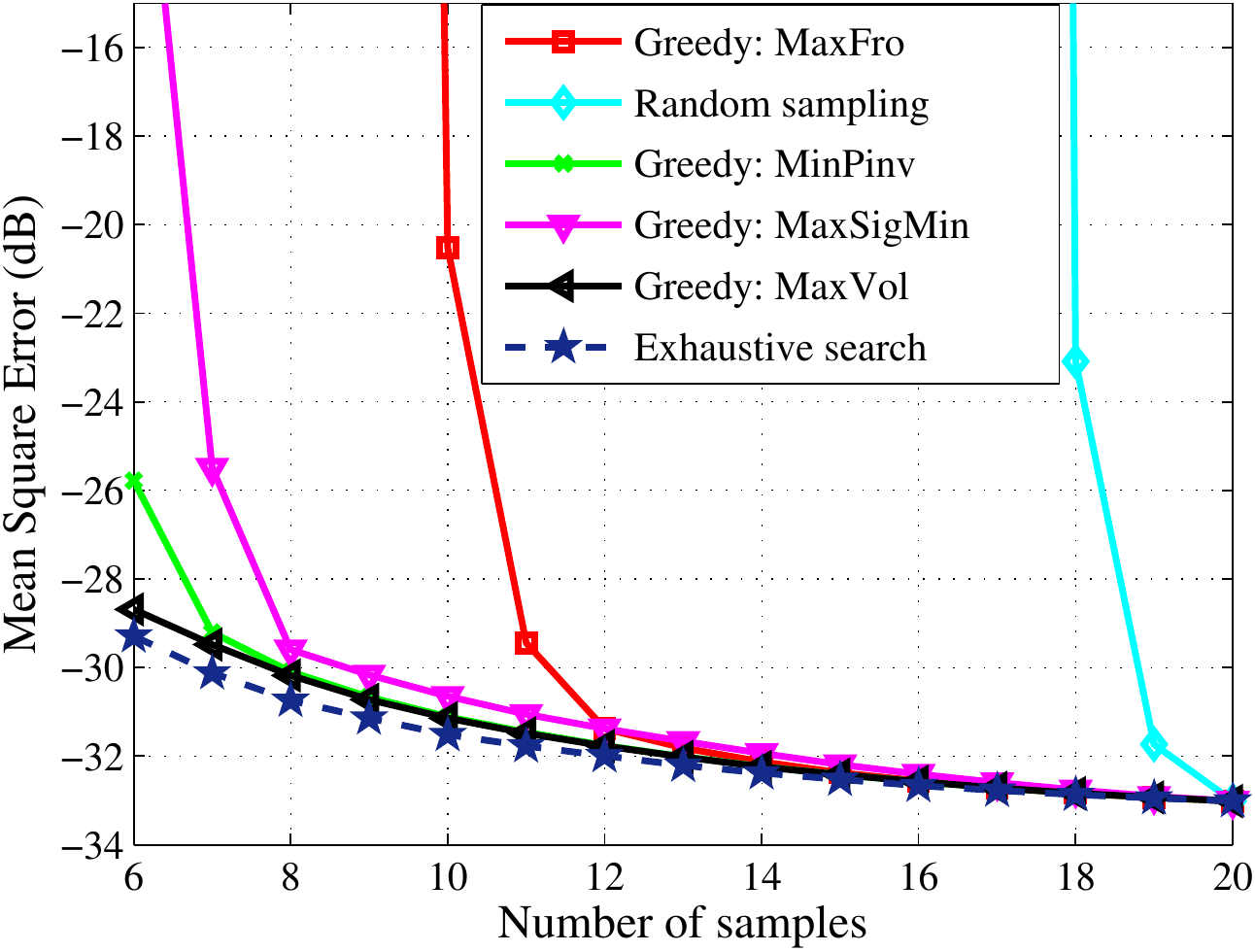}
\caption{}
\end{subfigure}
\hspace{.5cm}
\begin{subfigure}[t]{0.45\linewidth}
\includegraphics[width=\columnwidth,keepaspectratio]{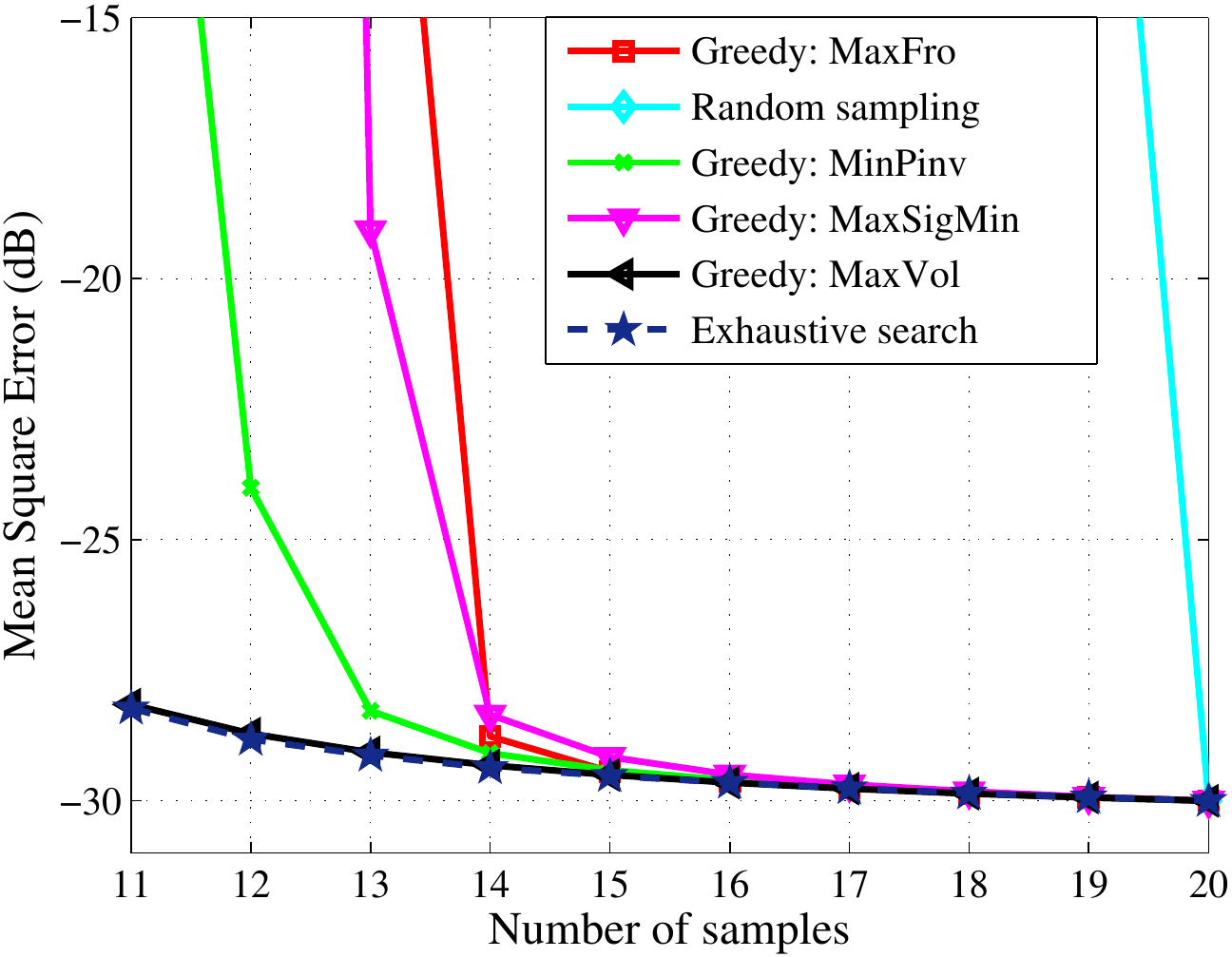}
\caption{}
\end{subfigure}
\caption{Behavior of Mean Squared Error versus number of samples, for different sampling strategies for a random geometric graph topology. (a) for the case of $\abs{\F} = 5$, whereas (b) for the case $\abs{\F} = 10$.}
\label{fig::Sampling_SF}
\end{figure*}

We compare now the performance obtained with the proposed sampling strategies, with random sampling and with the strategy recently  proposed in \cite{chen2015discrete} aimed at maximizing the minimum singular value of $\mSigma \mU^H \mD$. We consider, as an example, a random geometric graph model, with $N=20$ nodes randomly distributed over a unitary area, and having covering radius equal to 0.34. The corresponding results are shown in Fig. \ref{fig::Sampling_SF}, which reports the behavior of the mean square error (MSE) in (\ref{eq::expected_value_noise}) versus the number of samples. We consider band-limited signals with two different bandwidths $\abs{\F} = 5$ and $\abs{\F} = 10$. The observation model is given by (\ref{r=D(s+n)}) where the additive noise is generated as an uncorrelated, zero mean Gaussian random vector.  The results shown in the figures have been obtained by averaging over $500$ independent realizations of graph topologies. We compare five different sampling strategies, namely: (i) the random strategy, which picks up  nodes randomly; (ii) the greedy selection method of Algorithm 1, minimizing the Frobenius norm of $\left( \mSigma \mU^H \mD \right)^+$ (MinPinv); (iii) the greedy appraoch that maximizes the Frobenius norm of $\mSigma \mU^H \mD$ (MaxFro); (iv) the greedy selection method of Algorithm 2, that maximizes the volume of the parallelepiped built from the selected vectors (MaxVol);  and (v) the greedy algorithm  (MaxSigMin) maximizing the minimum singular value of $\mSigma \mU^H \mD$, recently proposed in \cite{chen2015discrete}. The performance of the globally optimal strategy obtained through an exhaustive search over all possible selections is also reported as a benchmark.

From Fig. \ref{fig::Sampling_SF} we observe that, as expected, as the number of samples increases, the mean squared error decreases. As a general remark, we can notice how random sampling can perform quite poorly. This poor results of random sampling emphasizes that, when sampling a graph signal, what matters is not only the number of samples, but also (and most important) {\it where} the samples are taken. Furthermore, we can notice how the proposed MaxVol strategy largely outperforms all the other strategies, while showing performance very close to the optimal combinatorial benchmark. Comparing the proposed MaxVol and MinPinv methods with the MaxSigMin approach, we see that the performance gain increases as the bandwidth increases. This happens because, as the bandwidth increases, more and more modes play a role in the final MSE, as opposed to the single mode associated to the minimum singular value.


\balance
\section{Conclusions}
In conclusion, in this paper we derived an uncertainty principle for graph signals, where the boundary of the admissible energy concentration in vertex and transformed domain is expressed in closed form. Then, we established a link between localization properties and sampling theory. Finally, we proposed a few alternative sampling strategies, based on greedy approaches, aimed to strike a good trade-off between performance and complexity. We compared the performance of the proposed methods with a very recently proposed method and with the globally optimal combinatorial search. Although sub-optimal, the MaxVol method exhibits very good performance, with only small losses with respect to the combinatorial search, at least in the case of random graph analyzed in this paper. Further investigations are taking place to assess the performance over different classes of graphs.

\bibliographystyle{MyIEEE}
\bibliography{refs}

\end{document}